\documentclass[english]{cccconf}

\usepackage[english]{babel}
\usepackage{theorem}
\theoremheaderfont{\itshape\bfseries}
{\theorembodyfont{\itshape}
	\newtheorem{assumption}{\textbf{Assumption}}
	\newtheorem{lemma}{\textbf{Lemma}}
	\newtheorem{definition}{\textbf{Definition}}
	\newtheorem{theorem}{\textbf{Theorem}}
	
	\newtheorem{remark}{\textbf{Remark}}
	
	\newtheorem{problem}{\textbf{Problem}}
	\newtheorem{solvability}{\textbf{Solvability Condition}}
}
\usepackage{graphicx}
\usepackage{newtxtext,newtxmath}
\usepackage{amsmath}
\usepackage{amsfonts}
\usepackage{mathrsfs}
\usepackage{xcolor}
\usepackage{graphicx}

\newcommand{\T}{^{\mbox{\tiny T}}}
\newcommand{\R}{\mathbb{R}}
\newcommand{\C}{\mathbb{C}}
\newcommand{\N}{\mathbb{N}}
\newcommand{\eps}{\varepsilon}
\let\leq\leqslant
\let\geq\geqslant

\newenvironment{proof}[1][Proof]%
{\par\noindent\textit{#1:\ }}%
{\hspace*{\fill} \rule{6pt}{6pt}}
\newenvironment{proof*}[1][Proof]%
{\par\noindent\textit{#1:\ }}{}

\DeclareMathOperator{\diag}{diag}

\DeclareMathOperator{\re}{Re}

\DeclareMathOperator{\rank}{rank}

\usepackage{tikz}
\usetikzlibrary{shapes,calc,arrows,patterns,decorations.pathmorphing
	,decorations.markings}
\usetikzlibrary{arrows.meta}

\newenvironment{system}[1]%
{\setlength{\arraycolsep}{0.5mm}\left\{ \; \begin{array}{#1}}%
	{\end{array} \right.}
\newenvironment{system*}[1]%
{\setlength{\arraycolsep}{0.5mm} \begin{array}{#1}}%
	{\end{array}}

\begin{document}

	\title{Output Synchronization of Heterogeneous Multi-agent Systems subject to Unknown, Non-uniform and Arbitrarily Large Input Delay: A Scale-free Protocol Design}
	\author{Donya Nojavanzadeh\aref{wsu}, Zhenwei Liu\aref{neu}, Ali Saberi\aref{wsu},
		Anton A. Stoorvogel\aref{ut}}

	\affiliation[wsu]{School of Electrical Engineering and Computer
		Science, Washington State University, Pullman, WA 99164, USA
		\email{donya.nojavanzadeh@wsu.edu; saberi@wsu.edu}}
		\affiliation[neu]{College of Information Science and
		Engineering, Northeastern University, Shenyang 110819, China
		\email{liuzhenwei@ise.neu.edu.cn}}
	\affiliation[ut]{Department of Electrical Engineering,
		Mathematics and Computer Science, University of Twente, Enschede, The Netherlands
		\email{A.A.Stoorvogel@utwente.nl}}

\title{\LARGE \textbf{ Scale-free Protocol Design for Output Synchronization of Heterogeneous Multi-agent subject to Unknown, Non-uniform and Arbitrarily Large Input Delays}}
	
	\maketitle

	\begin{abstract}
		This paper studies output synchronization problems for heterogeneous networks of continuous- or discrete-time right-invertible linear agents in presence of unknown, non-uniform and arbitrarily large input delay based on localized information exchange. It is assumed that all the agents are introspective, meaning that they have access to their own local measurements. Universal linear protocols are proposed for each agent to achieve output synchronizations. Proposed protocols are designed solely based on the agent models using no information about communication graph and the number of agents or other agent models information. Moreover, the protocols can tolerate arbitrarily large input delays.
	\end{abstract}
	
	\section{Introduction}
	
	Synchronization problem of multi-agent systems (MAS) has become a hot topic among researchers in recent years. Cooperative control of MAS is used in practical application such as robot network, autonomous vehicles, distributed sensor network, and others. The objective of synchronization is to secure an asymptotic agreement on a common state or output trajectory by local interaction among agents (see \cite{bai-arcak-wen,mesbahi-egerstedt,ren-book,wu-book} and references therein).
	
	Most of the attention has been focused on state synchronization of MAS, where each agent has access to a linear combination of its own state relative to that of the neighboring agents, called full-state coupling \cite{saber-murray3,saber-murray,saber-murray2,ren-atkins,ren-beard-atkins,tuna1}. A more realistic scenario which is partial-state coupling (i.e. agents share part of their information over the network) is studied in \cite{tuna2,li-duan-chen-huang,pogromsky-santoboni-nijmeijer,tuna3}.
	
	We identify two classes of MAS: homogeneous and heterogeneous. The agents dynamics can be different in heterogeneous networks. For heterogeneous network it is more reasonable to consider output synchronization since the dimensions of states and their physical interpretation may be different. Meanwhile a common assumption, especially for heterogeneous MAS is that agents are introspective; that is, agents possess some knowledge about their own states. There exist many results about this type of agents, for instance, homogenization based synchronization via local feedback \cite{kim-shim-seo,yang-saberi-stoorvogel-grip-journal},  $H_\infty$ design \cite{li-soh-xie-lewis-TAC2019}, HJB based optimal synchronization \cite{modares-lewis-kang-davoudi-TAC2018}, adaptation based event trigger regulated synchronization \cite{qian-liu-feng-TAC2019}, and a feed forward design for nonlinear agents \cite{chen-auto2019}.  Recently, scale-free collaborative protocol designs are developed for continuous-time heterogeneous MAS \cite{donya-liu-saberi-stoorvogel-ACC2020} and for homogeneous continues-time MAS subject to actuator saturation \cite{liu-saberi-stoorvogel-nojavanzadeh-cdc19}. 
	
	On the other hand, for non-introspective agents, regulated output synchronization for a heterogeneous network is studied in
	\cite{peymani-grip-saberi,peymani-grip-saberi-fossen}.  Other designs can also be found, such as an
	internal model principle based design
	\cite{wieland-sepulchre-allgower}, distributed high-gain observer
	based design \cite{grip-yang-saberi-stoorvogel-automatica}, 
	low-and-high gain based, purely distributed, linear time invariant protocol design \cite{grip-saberi-stoorvogel}.
	
	In practical applications, the network dynamics are not perfect and may be subject to delays. Time delays may afflict systems performance or even lead to instability. As it has been discussed in \cite{cao-yu-ren-chen}, two kinds of delay have been considered in the literature: input delay and communication delay. Input delay is the processing time to execute an input for each agent whereas communication delay can be considered as the time for transmitting information from origin agent to its destination.
	Some research have been done in the case of communication delay \cite{tian-liu,xiao-wang-tac,munz-papachristodoulou-allgower,munz-papachristodoulou-allgower2,lin-jia-auto,chopra-tac,chopra-spong-cdc06,stoorvogel-saberi-cdc16,ghabcheloo,klotz-obuz-kan}.  Regarding input delay, many efforts have been done (see \cite{ferrari-trecate,lin-jia-auto,lin-jia,tian-liu,xiao-wang,saber-murray2}) where they are mostly restricted to simple agent models such as first and second-order dynamics for both linear and nonlinear agents dynamics. Authors of \cite{wang-saberi-stoorvogel-grip-yang,wang-aut} studied state synchronization problems in the presence of unknown, uniform constant input delay for both continuous- and discrete-time networks with higher-order linear agents. \cite{zhang-saberi-stoorvogel-continues-discrete} has studied synchronization in homogeneous networks of both continuous- and discrete-time agents with unknown, non-uniform, and constant input delays.
	
	In this paper, we deal with output synchronization problem for heterogeneous MAS with continuous- or discrete-time introspective right-invertible agents in presence of unknown, non-uniform and arbitrarily large input delays. 
	Scale-free protocols are designed based on localized information exchange which do not require any knowledge of the communication network except connectivity. In particular, output synchronization of heterogeneous networks are achieved for an arbitrary number of agents. The protocol design is scale-free, namely,
\begin{itemize}
	\item  The design is independent of information about communication networks such as spectrum of the associated Laplacian matrix. That is to say, the universal dynamical protocols work for any communication network as long as it is connected.
	\item  The dynamic protocols are designed solely based on agent models and do not depend on communication network and the number of agents.
	\item The proposed protocols archive output synchronization for heterogeneous continuous- or discrete-time MAS with any number of agents, any unknown, non-uniform, input delays, and any communication network.
\end{itemize}

	\subsection*{Notations and definitions}
Given a matrix $A\in \mathbb{R}^{m\times n}$, $A\T$ denotes its
conjugate transpose and $\|A\|$ is the induced 2-norm while $\sigma_{\min}(A)$ denotes the smallest singular value of A. Let $j$ indicate $\sqrt{-1}$. A square matrix
$A$ is said to be Hurwitz stable if all its eigenvalues are in the open left half plane and is Schur stable if all its eigenvalues are inside the open unit disc. We denote by
$\diag\{A_1,\ldots, A_N \}$, a block-diagonal matrix with
$A_1,\ldots,A_N$ as the diagonal elements. $A\otimes B$ depicts the
Kronecker product between $A$ and $B$. $I_n$ denotes the
$n$-dimensional identity matrix and $0_n$ denotes $n\times n$ zero
matrix; sometimes we drop the subscript if the dimension is clear from
the context.
Notation $\overline{[t_1, t_2]}$ means $\overline{[t_1, t_2]} = \{ t \in \mathbb{Z} : t_1 \leq t \leq t_2\}.$

To describe the information flow among the agents we associate a \emph{weighted graph} $\mathcal{G}$ to the communication network. The weighted graph $\mathcal{G}$ is defined by a triple
$(\mathcal{V}, \mathcal{E}, \mathcal{A})$ where
$\mathcal{V}=\{1,\ldots, N\}$ is a node set, $\mathcal{E}$ is a set of
pairs of nodes indicating connections among nodes, and
$\mathcal{A}=[a_{ij}]\in \mathbb{R}^{N\times N}$ is the weighted adjacency matrix with non negative elements $a_{ij}$. Each pair in $\mathcal{E}$ is called an \emph{edge}, where
$a_{ij}>0$ denotes an edge $(j,i)\in \mathcal{E}$ from node $j$ to
node $i$ with weight $a_{ij}$. Moreover, $a_{ij}=0$ if there is no
edge from node $j$ to node $i$. We assume there are no self-loops,
i.e.\ we have $a_{ii}=0$. A \emph{path} from node $i_1$ to $i_k$ is a
sequence of nodes $\{i_1,\ldots, i_k\}$ such that
$(i_j, i_{j+1})\in \mathcal{E}$ for $j=1,\ldots, k-1$. A \emph{directed tree} is a sub-graph (subset
of nodes and edges) in which every node has exactly one parent node except for one node, called the \emph{root}, which has no parent node. The \emph{root set} is the set of root nodes. A \emph{directed spanning tree} is a sub-graph which is
a directed tree containing all the nodes of the original graph. If a directed spanning tree exists, the root has a directed path to every other node in the tree.  

For a weighted graph $\mathcal{G}$, the matrix
$L=[\ell_{ij}]$ with
\[
\ell_{ij}=
\begin{system}{cl}
\sum_{k=1}^{N} a_{ik}, & i=j,\\
-a_{ij}, & i\neq j,
\end{system}
\]
is called the \emph{Laplacian matrix} associated with the graph
$\mathcal{G}$. The Laplacian matrix $L$ has all its eigenvalues in the
closed right half plane and at least one eigenvalue at zero associated
with right eigenvector $\textbf{1}$ \cite{royle-godsil}. Moreover, if the graph contains a directed spanning tree, the Laplacian matrix $L$ has a single eigenvalue at the origin and all other eigenvalues are located in the open right-half complex plane \cite{ren-book}.

%
%

\section{Problem Formulation}

Consider a MAS consisting of $N$ non-identical linear agents:
\begin{equation}\label{hete_sys}
\begin{cases}
s{x}_i(t)=A_ix_i(t)+B_iu_i(t),\\
y_i(t)=C_ix_i(t),\\
\end{cases}
\end{equation}
where $x_i\in\mathbb{R}^{n_i}$, $u_i\in\mathbb{R}^{m_i}$ and $y_i\in\mathbb{R}^p$ are the state, input, output of agent $i$ for $i=1,\ldots, N$. In the aforementioned presentation, for continuous-time systems, $s$ denotes the time derivative, i.e., $s x_i(t) = \dot{x}_i(t)$ for $t \in \mathbb{R}$; while for discrete-time systems, $s$ denotes the time shift, i.e., $s x_i(t) = x_i(t + 1)$ for $t \in \mathbb{Z}$.


The agents are introspective, meaning that each agent collects a local measurement $z_i\in \mathbb{R}^{q_i}$ of its internal dynamics. In other words, each agent has access to the quantity
\begin{equation}\label{local}
	z_i(t)=C_i^mx_i(t).
\end{equation}
where $z_i(t)\in \mathbb{R}^{q_i}$.

 We define the set of graphs $\mathbb{G}^N$ for the network communication topology as following.

\begin{definition}\label{def1}
	Let $\mathbb{G}^N$ denote the set of directed graphs of $N$ agents which contain a directed spanning tree.	
\end{definition}

The communication network provides agent $i$ with the following information which is a linear combination of its own output relative to that of other agents:

\begin{equation}\label{zeta}
\zeta_i(t)=\sum_{j=1}^{N}a_{ij}(y_i(t)-y_j(t)) 
\end{equation}

where $a_{ij}>0$ and $a_{ii}=0$, and $\mathcal{A}=[a_{ij}]$ is the weighting matrix associated to the network graph
$\mathcal{G}$. We refer to this case as partial-state coupling since only part of the states is exchanged among the agents. When $C=I$, it means all states are communicated over the network and we call it full-state coupling.  Basically, output synchronization is considered for heterogeneous MAS, therefore, we focus on partial-state coupling.

In this paper, we also introduce a localized
information exchange among protocols. In particular, each agent 
$i=1,\ldots, N$ has access to the localized information, denoted by
$\hat{\zeta}_i(t)$, of the form
\begin{equation}\label{zetahat}
\hat{\zeta}_i(t)=\sum_{j=1}^{N}a_{ij}(\eta_i(t)-\eta_j(t))
\end{equation}
where $\eta_j(t)\in\mathbb{R}^n$ is a variable produced internally by agent $j$ and to be defined in next sections.

In the case of networks with discrete-time agents, for any graph $\mathcal{G}\in \mathbb{G}^N$, with the associated Laplacian matrix $L$, we define
\begin{equation}\label{hodt-LDa}
{D}=I-(I+D_{\text{in}})^{-1}{L}
\end{equation}
where
\begin{equation}\label{d-in}
D_{\text{in}}=\diag\{d_{in}(i) \}
\end{equation}
with $d_{in}(i)=\sum_{j=1}^N a_{ij}$. The weight matrix $D=[d_{ij}]$ is a so-called, row stochastic matrix, where $d_{ij}\geq 0$, and we choose
$d_{ii}=1-\sum_{j=1,j\neq i}^Nd_{ij}$ such that $\sum_{j=1}^Nd_{ij}=1$ for $i,j\in\{1,\ldots,N\}$. Note that $d_{ii}$ satisfies $d_{ii}>0$.

Therefore, for discrete-time networks we can obtain the information exchange as
\begin{equation}\label{zeta-d}
{\zeta}_i(t)=\sum_{j=1}^Nd_{ij}(y_i(t)-y_j(t))
\end{equation}
and we can rewrite $\hat{\zeta}_i(t)$ as
\begin{equation}\label{zetahat-d}
\hat{\zeta}_i(t)=\sum_{j=1}^Nd_{ij}(\eta_i(t)-\eta_j(t))
\end{equation}


The heterogeneous MAS is said to achieve output synchronization if 
\begin{equation}\label{synch_out}
\lim\limits_{t\to\infty}(y_i(t)-y_j(t))=0, \quad\text{for $i,j \in \{1, \dots ,N\}$}.
\end{equation}

In this paper, we introduce a protocol architecture as shown below in Figure \ref{Network}.
\begin{figure}[h]
	\includegraphics[width=8.3cm, height=4.5cm]{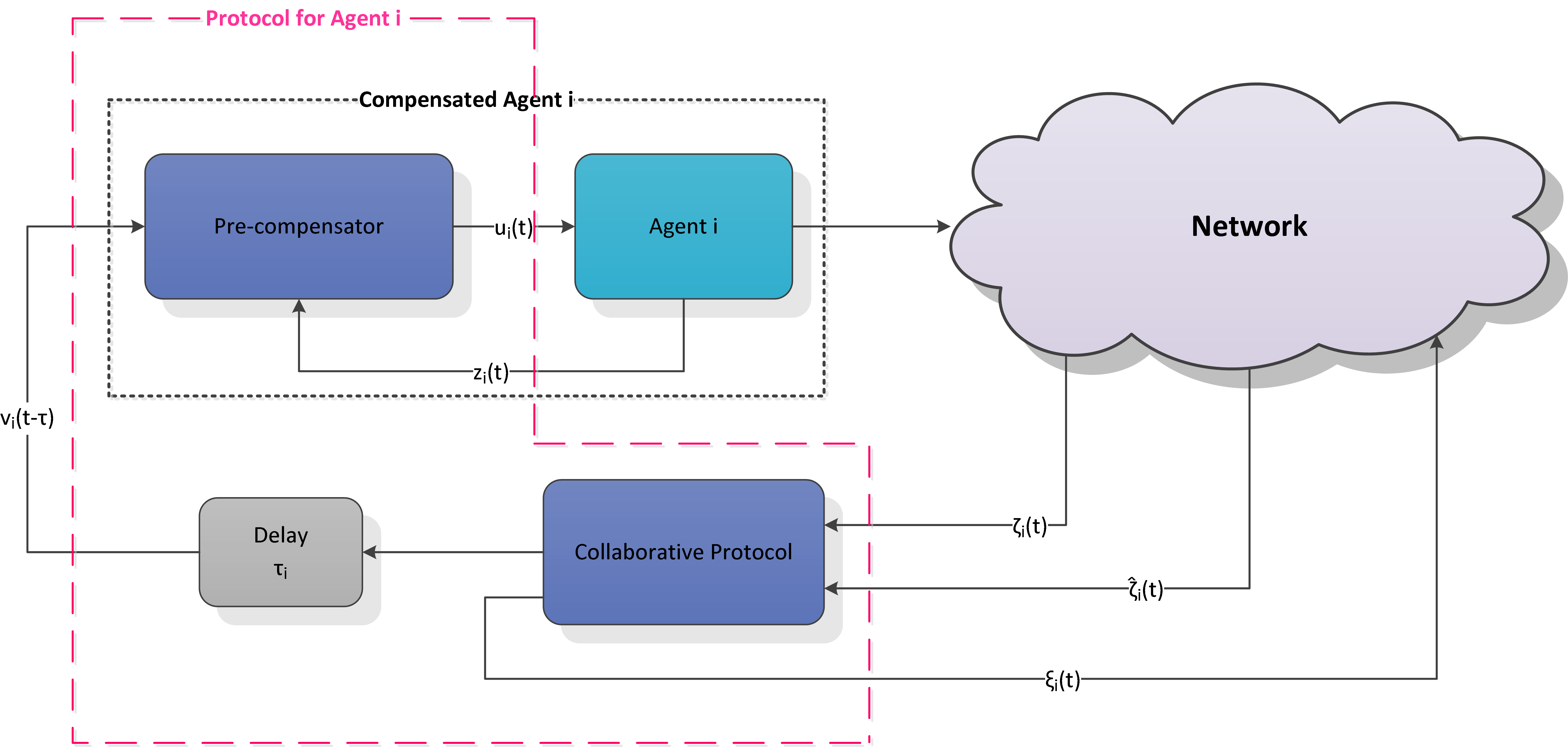}
	\centering
	\caption{Architecture of the protocol}\label{Network}
\end{figure}

As seen in the above figure our protocol has two main modules.
\begin{enumerate}
	\item precompensator
	\item collaborative protocol
\end{enumerate}
we also assume that agents are experiencing input delay $\tau_i$ represented by the delay block in the figure. Note that as discussed in \cite{cao-yu-ren-chen}  input delay can be considered as the summation of computation time and execution time. We formulate output synchronization problem for continuous-or discrete-time heterogeneous networks as follows.

\begin{problem}\label{prob_sync}
	Consider a heterogeneous network of $N$ agents \eqref{hete_sys} with local information \eqref{local} and a given $\bar{\tau}$, where $\tau_i<\bar{\tau}$ for $i=1,\hdots, N$. Let the associated network communication be given by \eqref{zeta} and \eqref{zetahat} for continuous-time and by \eqref{zeta-d} and \eqref{zetahat-d} for discrete-time networks. Let $\mathbb{G}^N$ be the set of network graphs as defined in Definition \ref{def1}. 
	
	The \textbf{scalable output synchronization problem based on localized information exchange in presence of input delay} is to find, if possible, a linear dynamic controller for each agent $i \in\{1, \dots, N\}$, as illustrated in Figure \ref{Network}, using only knowledge of agent models, i.e. $(C_i,A_i,B_i)$, and upper bound on delays $\bar{\tau}$.


Then, output synchronization \eqref{synch_out} is achieved for any $N$ and any graph $\mathcal{G}\in\mathbb{G}^N$.
\end{problem}

\begin{remark}
Note that in our problem formulation it is embedded that our linear dynamic protocols are designed only based on agent models $(C_i, A_i,B_i)$ and given  upper bound on input delays $\bar{\tau}$. Moreover, this universal protocol is scale-free, meaning that it works for any network with any number of agents, as long as the associated communication graph contains a spanning tree.
\end{remark}

We make the following assumption for the agents.
\begin{assumption}\label{ass2}
	For agents $i \in \{1,\dots,N\}$, 
	\begin{enumerate}
		\item $(A_i,B_i)$ is stabilizable.
		\item $(C_i, A_i)$ is detectable.
		\item $(C_i,A_i,B_i)$ is right-invertible
		\item $(C_i^m,A_i)$ is detectable. 
	\end{enumerate}
\end{assumption}
Then we have the following solvability condition for scalable output synchronization of heterogeneous continuous- or discrete-time MAS in presence of input delay.

\begin{solvability}\label{thm_out_syn}
	Consider a heterogeneous network of $N$ agents \eqref{hete_sys} with local information \eqref{local} satisfying Assumption \ref{ass2} and a given $\bar{\tau}$. Let the associated network communication be given by \eqref{zeta} and \eqref{zeta-d} for continuous- and discrete-time MAS respectively. Let $\mathbb{G}^N$ be the set of network graphs as defined in Definition \ref{def1}. 
	
	Then, the scalable output synchronization problem based on localized information exchange in presence of input delay as defined in Problem \ref{prob_sync} is solvable.
\end{solvability}
\begin{proof}
	Proof will be given in next sections.
	\end{proof}

\section{Protocol Design}\label{OS}
In this section, we design dynamic protocols to solve scalable output synchronization problem for heterogeneous networks of continuoue- or discrete-time agents. The architecture of the proposed protocol is illustrated in Figure \ref{Network}.

\subsection*{Architecture of the protocol}
Our design methodology consists of two modules as shown in Figure \ref{Network}.
\begin{enumerate}
	\item The first module is to reshape the dynamics of the agents to obtain the target model by designing precompensators following our previous results stated in \cite{yang-saberi-stoorvogel-grip-journal,wang-saberi-yang}.
	\item The second module is designing collaborate protocols for almost homogenized agents to achieve output synchronization in presence of input delay. 
\end{enumerate}

\subsection{Designing pre-compensators}
In this section, the goal of the design is to reshape the agent models and obtain suitable target model i.e. $(C, A, B)$ such that following conditions are satisfied.
\begin{enumerate}
	\item $\rank(C)=p$.
	\item $(C, A, B)$ is invertible of uniform rank $n_q\ge\bar{n}_d$, and has no invariant zeros.
	\item eigenvalues of $A$ are in closed left half plane (closed unit disc for discrete-time systems).
	\item eigenvalues of $A$ satisfy the additional property 
		\begin{equation}
		\bar{\tau}\omega_{\max}<\frac{\pi}{2},
		\end{equation}
		where $\omega_{\max}$ is defined as
	\begin{equation}\label{omega-max}
	\omega_{\max}=
	\begin{cases}
	0, \qquad \text{A is Hurwitz stable},&\\
	\max\{\omega\in\R | \det(j\omega I-A)=0\}, &\text{otherwise}
	\end{cases}
	\end{equation}
	for continuous-time systems, and
	\begin{equation}\label{omega-max-d}
	\omega_{\max}=
	\begin{cases}
	0, \qquad\text{A is Schur stable},& \\
	\max\{\omega\in[0,\pi] | \det(e^{j\omega} I-A)=0\}, &\text{otherwise}
	\end{cases}
	\end{equation}
	for discrete-time systems.		
\end{enumerate}

\begin{remark}
	We would like to make several observations:
	\begin{enumerate}
		\item The property that the triple $(C, A, B)$ is invertible and has no invariant zero implies that $(A, B)$ is controllable and $(C, A)$ is observable.
		\item The triple $(C, A, B)$ is arbitrarily assignable as long as the conditions are satisfied. In this paper, $A$ is chosen such that its eigenvalues are in closed left half plane (in closed unit disc for discrete-time systems). For example, one can choose $A$ such that $A\T+A=0$ for continuous-time systems and $A\T A=I$ for discrete-time systems.
	\end{enumerate}
\end{remark}

Next, given chosen $(C,A,B)$, by utilizing the design methodologies from \cite[Appendix B]{yang-saberi-stoorvogel-grip-journal}  for continuous- time and \cite[Appendix A.1]{wang-saberi-yang} for discrete-time systems, we design a pre-compensator for each agent $i \in \{1, \dots, N\}$, of the form
	\begin{equation}\label{pre_con}
	\begin{system}{cl}
	s {\xi}_i(t)&=A_{i,h}\xi_i(t)+B_{i,h}z_i(t)+E_{i,h}v_i(t),\\
	u_i(t)&=C_{i,h}\xi_i(t)+D_{i,h}v_i(t),
	\end{system}
	\end{equation} 
	which yields the compensated agents as
	\begin{equation}\label{sys_homo}
	\begin{system*}{cl}
	s {\bar{x}}_i(t)&=A\bar{x}_i(t)+B(v_i(t)+\psi_i(t)),\\
	{y}_i(t)&=C\bar{x}_i(t),
	\end{system*}
	\end{equation} 
	where $\psi_i(t)$ is given by 
	\begin{equation}\label{sys-rho}
	\begin{system*}{cl}
	s {\omega}_i(t)&=A_{i,s}\omega_i(t),\\
	\psi_i(t)&=C_{i,s}\omega_i(t),
	\end{system*}
	\end{equation}
	and $A_{i,s}$ is Hurwitz stable (Schur stable for discrete-time systems). Figure \ref{Network2} shows the compensated agents as a component of Figure \ref{Network}.
	
			\begin{figure}[h]
		\includegraphics[width=6cm, height=2.5cm]{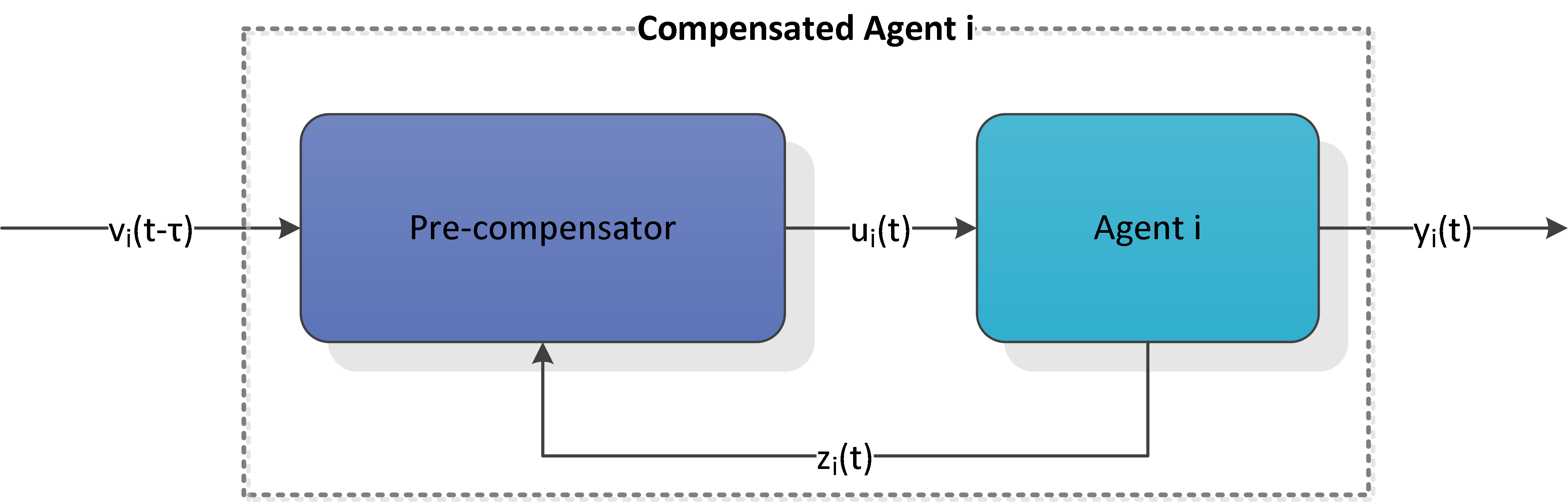}
		\centering
		\caption{Compensated agent}\label{Network2}
	\end{figure}

 Note that the compensated agents are homogenized and have the target models $(C, A, B)$. 
 
 We consider our compensated agents in presence of input delay as
\begin{equation}\label{sys_homo-d-con}
\begin{system*}{cl}
s{\bar{x}}_i(t)&=A\bar{x}_i(t)+B(v_i(t-\tau_i)+\psi_i(t)),\\
{y}_i(t)&=C\bar{x}_i(t),
\end{system*}
\end{equation} 
where $\psi_i(t)$ is given by \eqref{sys-rho}.

\subsection{Designing collaborative protocols for compensated agents} 
In this subsection, to achieve output synchronization, we design collaborative protocols for almost homogenized continuous- or discrete-time  agents in presence of unknown, non-uniform and arbitrarily large input delays.

\subsubsection{Continuous-time MAS}\label{Continuous}
  We design a dynamic protocol based on localized information exchange as
\begin{equation}\label{pscp-con}
\begin{system}{cll}
\dot{\hat{x}}_i(t) &=& A\hat{x}_i(t)+B\hat{\zeta}_{i2}(t)+F(\zeta_i(t)-C\hat{x}_i(t))\\
\dot{\chi}_i(t) &=& A\chi_i(t)+Bv_i(t-\tau_i)+\hat{x}_i(t)-\hat{\zeta}_{i1}(t)\\
v_i(t) &=& -\rho B\T P_{\eps}\chi_i(t),
\end{system}
\end{equation}
for $i=1,\ldots,N$ where $F$ is a pre-design matrix such that $A-FC$ is Hurwitz stable and $\rho>0$. $\eps$ is a parameter satisfying $\eps\in (0,1]$, $P_{\eps}$ satisfies
\begin{equation}\label{arespecial-con}
A\T P_{\eps} + P_{\eps} A -  P_{\eps} BB\T
P_{\eps} + \eps I = 0 
\end{equation}
note that for any $\eps>0$, there exists a unique solution of \eqref{arespecial-con}.
\begin{remark}
	\eqref{arespecial-con} is a special case of the general low-gain $H_2$ algebraic Riccati equation ($H_2$-ARE), which is written as follows:
	\begin{equation}\label{aregeneral-con}
A\T P_{\eps} + P_{\eps} A -  P_{\eps} BR_\eps^{-1} B\T
P_{\eps} + Q_\eps = 0 
	\end{equation}
	where $R_{\eps} > 0$, and $Q_{\eps} > 0$ is such that $Q_{\eps} \rightarrow 0$ as $\eps \rightarrow 0.$ In our case, we restrict our attention to $Q_{\eps} = \eps I$ and $R_{\eps} = I.$ However, as shown in \cite{saberi-stoorvogel-sannuti-exter}, when $A$ is neutrally stable, there exists a suitable (nontrivial) choice of $Q_{\eps}$ and $R_{\eps}$ which yields an explicit solution of \eqref{aregeneral-con}, of the form
	\begin{equation}\label{neutral}
	P_{\eps} = \eps P
	\end{equation}
	where $P$ is a positive definite matrix that satisfies $A^T P +PA \leq 0.$
\end{remark}

The agents communicate $\eta_i=(\eta_{i1}\T,\eta_{i2}\T)\T$ where $\eta_{i1}(t)=\chi_i(t)$ and $\eta_{i2}(t)=v_i(t-\tau_i)$, therefore each agent has access to the localized information $\hat{\zeta}_i=(\hat{\zeta}_{i1}\T,\hat{\zeta}_{i2}\T)\T$:
\begin{equation}\label{add_1-con}
\hat{\zeta}_{i1}(t)=\sum_{j=1}^Na_{ij}(\chi_i(t)-\chi_j(t)),
\end{equation}
and
\begin{equation}\label{add_2-con}
\hat{\zeta}_{i2}(t)=\sum_{j=1}^{N}a_{ij}(v_i(t-\tau_i)-v_j(t-\tau_j)).
\end{equation}
while $\zeta_i(t)$ is defined by \eqref{zeta}. Finally, we combine the designed protocol for homogenized network with pre-compensators and present our protocols as:
\begin{equation}\label{pscp-con-final}
\begin{system}{cl}
\dot{\xi}_i(t)&=A_{i,h}\xi_i(t)+B_{i,h}z_i(t)-\rho E_{i,h}B\T P_\eps\chi_i(t),\\
\dot{\hat{x}}_i(t) &= A\hat{x}_i(t)+B\hat{\zeta}_{i2}(t)+F(\zeta_i(t)-C\hat{x}_i(t))\\
\dot{\chi}_i(t) &= A\chi_i(t)-\rho BB\T P_\eps \chi_i(t-\tau_i)+\hat{x}_i(t)-\hat{\zeta}_{i1}(t)\\
u_i(t)&=C_{i,h}\xi_i-\rho D_{i,h}B\T P_\eps \chi_i,
\end{system}
\end{equation}

Then, we have the following theorem for \emph{scalable} output synchronization of heterogeneous continuous-time MAS in presence of input saturation.

\begin{theorem}\label{thm-con}
	Consider a heterogeneous network of $N$ agents \eqref{hete_sys} with local information \eqref{local} satisfying Assumption \ref{ass2} and a given $\bar{\tau}$. Let the associated network communication be given by \eqref{zeta}. Let $\mathbb{G}^N$ be the set of network graphs as defined in Definition \ref{def1}.
	
	Then the scalable output synchronization problem as stated in Problem \ref{prob_sync} is solvable.
	In particular, there exist a $\rho^*>0.5$ and for any fixed $\rho>\rho^*$, there exists an $\eps^*$ such that for any $\eps \in (0,\eps^*]$, dynamic protocol given by \eqref{pscp-con-final} and \eqref{pscp-con} solves the scalable output synchronization problem for any $N$ and any graph
	$\mathcal{G}\in\mathbb{G}^N$.
\end{theorem}
\begin{proof}[Proof of Theorem \ref{thm-con}] Let $\bar{x}_i^o=\bar{x}_i-\bar{x}_N,\
	y_i^o=y_i-y_N,\
	\hat{x}_i^o=\hat{x}_i-\hat{x}_N, \text{ and }
	\chi_i^o=\chi_i-\chi_N.$
	Then, we have 
	\[
	\begin{system*}{ll}
	\dot{\bar{x}}_i^o(t)&=A\bar{x}_i^o(t)+B(v_i(t-\tau_i)-v_N(t-\tau_N)+\psi_i(t)-\psi_N(t)),\\
	{y}_i^o(t)&=C\bar{x}_i^o(t),\\
	\bar{\zeta}_i(t)&=\zeta_i(t)-\zeta_N(t)=\sum_{j=1}^{N-1}\bar{\ell}_{ij}{y}_j^o(t),\\
	\dot{\hat{x}}_i^o(t)&=A\hat{x}_i^o+B(\hat{\zeta}_{i2}(t)-\hat{\zeta}_{N2}(t))+F(\bar{\zeta}_i-C\hat{x}_i^o)\\
	\dot{\chi}_i^o(t)&=A\chi_i^o+B(v_i(t-\tau_i)-v_N(t-\tau_N))+\hat{x}_i^o(t)-\sum_{j=1}^{N-1}\bar{\ell}_{ij}{\chi}_j^o\\
	\end{system*}
	\]
	where $\bar{\ell}_{ij}=\ell_{ij}-\ell_{Nj}$ for $i,j=1,\cdots,N-1$. Note that eigenvalues of 
	\[
	\bar{L}=[\bar{\ell}_{ij}]_{(N-1)\times(N-1)}
	\]
	are equal to the nonzero eigenvalues of $L$ (see \cite{zhang-saberi-stoorvogel-delay}).
	
	We define
	\[
	\begin{system*}{ll}
	\tilde{x}(t)&=\begin{pmatrix}
	\bar{x}_1^o(t)\\ \vdots\\ \bar{x}_{N-1}^o(t)
	\end{pmatrix},\hat{x}(t)=\begin{pmatrix}
	\hat{x}_1^o(t)\\ \vdots\\ \hat{x}_{N-1}^o(t)
	\end{pmatrix},\chi(t)=\begin{pmatrix}
	\chi_1^o(t)\\ \vdots\\ \chi_{N-1}^o(t)
	\end{pmatrix},\\
	\bar{x}^\tau(t)&=\begin{pmatrix}
	\bar{x}_1^o(t-\tau_1)\\ \vdots\\ \bar{x}_{N-1}^o(t-\tau_{N-1})
	\end{pmatrix},\chi^\tau(t)=\begin{pmatrix}
	\chi_1^o(t-\tau_1)\\ \vdots\\ \chi_{N-1}^o(t-\tau_{N-1})
	\end{pmatrix},\\
	\psi(t)&=\begin{pmatrix}
	\psi_1(t)\\ \vdots\\ \psi_N(t)\end{pmatrix},\omega(t)=\begin{pmatrix}
	\omega_1(t)\\ \vdots\\ \omega_N(t)\end{pmatrix}.
	\end{system*} 
	\]
	
	Then we have the following closed-loop system
	\begin{equation}
	\begin{system*}{l}
	\dot{\tilde{x}}(t)=(I\otimes A) \tilde{x}(t)-\rho (I\otimes BB\T P_\eps)\chi^\tau(t)+(\Pi\otimes B)\psi(t)\\
	\dot{\hat{x}}(t) = I\otimes (A-FC)\hat{x}(t)-\rho(\bar{L}\otimes B B\T P_\eps)\chi^\tau(t)+(\bar{L}\otimes KC)\tilde{x}(t) \\
	\dot{\chi} (t)= (I\otimes A-\bar{L}\otimes I)\chi(t)-\rho(I\otimes BB\T P_\eps)\chi^\tau(t)+\hat{x}(t)
	\end{system*}
	\end{equation}
	where $\Pi=\begin{pmatrix}
	I&-\mathbf{1}
	\end{pmatrix}$. By defining $e(t)=\tilde{x}(t)-\chi(t)$ and $\bar{e}(t)=(\bar{L}\otimes I)\tilde{x}(t)-\hat{x}(t)$, we can obtain  
	\begin{equation}\label{newsystem3}
	\begin{system*}{ll}
	\dot{\tilde{x}}(t)=&(I\otimes A) \tilde{x}(t)-\rho(I\otimes BB\T P_\eps)\tilde{x}^\tau(t)+\rho(I\otimes BB\T P_\eps)e^\tau(t)\\
	&\hspace{4.7cm}+(\Pi\otimes B)C_s\omega(t)\\
	\dot{\bar{e}}(t)=&I\otimes (A-FC)\bar{e}(t)+(\bar{L}\Pi\otimes B) C_s\omega(t)\\
	\dot{e}(t)=&(I\otimes A-\bar{L}\otimes I)e(t)+\bar{e}(t)+(\Pi\otimes B) C_s\omega(t)
	\end{system*}
	\end{equation}
	where $e^\tau(t)=\tilde{x}^\tau(t)-\chi^\tau(t)$. We use critical lemma \ref{hode-lemma-ddelay-con} for the stability of delayed system. The proof proceeds in two steps. 
	
	\textbf{Step c.1:}  First, we prove the stability of system \eqref{newsystem3} without delay. By combining \eqref{newsystem3} and \eqref{sys-rho}, we have
	\begin{multline}\label{newsystem33}
	\begin{pmatrix}
	\dot{\tilde{x}}\\\dot{\bar{e}}\\\dot{e}\\\dot{\omega}
	\end{pmatrix}=\left(\begin{array}{cc}
	I\otimes A-\rho (I\otimes BB\T P_\eps)&0\\
	0&I\otimes (A-FC)\\
	0&I\\
	0&0
	\end{array}\right.
	\\
	\left.\begin{array}{cc}
	\rho (I\otimes BB\T P_\eps)&(\Pi\otimes B) C_s\\
	0&(\bar{L}\Pi\otimes B) C_s\\
	I\otimes A-\bar{L}\otimes I&(\Pi\otimes B) C_s\\
	0&A_s
	\end{array}\right)
	\begin{pmatrix}
	\bar{x}\\e\\\bar{e}\\\omega
	\end{pmatrix}
	\end{multline}
	Since all eigenvalues of $\bar{L}$ are positive, we have
	\begin{equation}\label{boundapl}
	(T\otimes I)(I\otimes A-\bar{L}\otimes I)(T^{-1}\otimes I)=I\otimes A-\bar{J}\otimes I
	\end{equation}
	for a non-singular transformation matrix $T$, where	\eqref{boundapl}  is upper triangular Jordan form with $A-\lambda_i I$ for $i=1,\cdots,N-1$ on the diagonal. Since all eigenvalues of $A$ are in the closed left half plane, $A-\lambda_i I$ is stable. Therefore, all eigenvalues of $I\otimes A-\bar{L}\otimes I$ have negative real part.
	Therefore, we have that the dynamics for $e$ is asymptotically stable. Meanwhile, since $A-FC$ is Hurwitz stable, one can obtain
	\[
	\lim_{t\to \infty}\bar{e}(t)\to 0 \text{ and }\lim_{t\to \infty}e(t)\to 0
	\]
	i.e. we just need to prove the stability of 
	\[		
	\dot{\tilde{x}}_i(t)=(A-\rho BB\T P_\eps)\tilde{x}_i(t).
	\]
	
	Based on lemma \ref{low-gain} , $A- \rho BB\T P_\eps$ is Hurwitz stable for $\eps >0$ and $\rho >0.5$.
	
	\textbf{Step c.2:} In this step, since we have that $e_i$ and $\bar{e}_i$ are asymptotically stable, we just need to prove the stability of
	\[
	\dot{\tilde{x}}_i(t)= A \tilde{x}_i(t)- \rho BB\T P_\eps\tilde{x}_i(t-\tau_i)
	\]
	for $i=1,\hdots, N$. Similar to the proof of \cite[Theorem 1]{consensus-identical-delay-c}, there exists an $\eps^*$ only function of $(C,A,B)$ such that by choosing 
	\[
	\rho>\frac{1}{2\cos(\bar{\tau}\omega_{\max} )}.
	\] 	
	we can obtain the synchronization result for any $\eps \in (0,\eps^*]$.
\end{proof}
\subsubsection{Discrete-time MAS}\label{Discrete}
In this subsection, we design dynamic protocols with localized information exchange for \eqref{sys_homo} and \eqref{sys-rho} as
\begin{equation}\label{pscp3}
\begin{system}{cll}
\hat{x}_i(t+1) &=& A\hat{x}_i(t)+B\hat{\zeta}_{i2}(t)+F({\zeta}_i(t)-C\hat{x}_i(t)) \\
\chi_i(t+1) &=& A\chi_i(t)+Bv_i(t-\tau_i)+A\hat{x}_i(t)-A\hat{\zeta}_{i1}(t)\\
v_i(t) &=& -\rho K_{\eps}\chi_i(t),
\end{system}
\end{equation}
for $i=1,\ldots,N$ where $F$ is a pre-design matrix such that $A-FC$ is Schur stable, and 
\[
K_\eps=(I+B\T P_\eps B)^{-1}B\T P_\eps A,
\]
and $\rho>0$. $\eps$ is a parameter satisfying
$\eps\in (0,1]$, where for any $\eps>0$, $P_{\eps}$ is the unique solution of 
\begin{equation}\label{arespecial-dis}
A\T P_{\eps}A - P_{\eps} -  A\T P_{\eps}B(I +B\T P_{\eps} B)^{-1}B\T
P_{\eps} A +  \eps I = 0 
\end{equation}
\begin{remark}
	\eqref{arespecial-dis} is a special case of the general low-gain $H_2$ discrete algebraic Riccati equation ($H_2$-DARE), which is written as follows:
	\begin{equation}\label{aregeneral}
	A\T P_{\eps}A - P_{\eps} -  A\T P_{\eps}B(R_{\eps}+B\T P_{\eps} B)^{-1}B\T
	P_{\eps} A +  Q_{\eps} = 0 
	\end{equation}
	where $R_{\eps} > 0$, and $Q_{\eps} > 0$ is such that $Q_{\eps} \rightarrow 0$ as $\eps \rightarrow 0.$ In our case, we restrict our attention to $Q_{\eps} = \eps I$ and $R_{\eps} = I.$ However, as shown in \cite{saberi-stoorvogel-sannuti-exter}, when $A$ is neutrally stable, there exists a suitable (nontrivial) choice of $Q_{\eps}$ and $R_{\eps}$ which yields an explicit solution of \eqref{aregeneral}, of the form
	\begin{equation}\label{neutral-d}
	P_{\eps} = \eps P
	\end{equation}
	where $P$ is a positive definite matrix that satisfies $A^T P A \leq P.$
\end{remark}
 In this protocol, the agents communicate $\eta_i=(\eta_{i1}\T,\eta_{i2}\T)\T$ where $\eta_{i1}(t)=\chi_i(t)$ and $\eta_{i2}(t)=v_i(t-\tau_i)$, therefore each agent has access to the localized information $\hat{\zeta}_i=(\hat{\zeta}_{i1}\T,\hat{\zeta}_{i2}\T)\T$:

\begin{equation}\label{add_1}
\hat{\zeta}_{i1}(t)=\sum_{j=1}^N d_{ij}(\chi_i(t)-\chi_j(t)),
\end{equation}
and
\begin{equation}\label{add_2}
\hat{\zeta}_{i2}(t)=\sum_{j=1}^{N} d_{ij}(v_i(t-\tau_i)-v_j(t-\tau_j)).
\end{equation}
while ${\zeta}_i(t)$ is defined by \eqref{zeta-d}. Finally, we combine the designed protocol for homogenized network with pre-compensators designed in step $1$ to get our protocol as:
\begin{equation}\label{pscp1final}
\begin{system}{cl}
{\xi}_i(t+1)&=A_{i,h}\xi_i(t)+B_{i,h}z_i(t)-\rho E_{i,h}K_\eps\chi_i(t),\\
{\hat{x}}_i(t+1) &= A\hat{x}_i(t)+B\hat{\zeta}_{i2}(t)+F(\zeta_i(t)-C\hat{x}_i(t))\\
{\chi}_i(t+1) &= A\chi_i(t)-\rho BK_\eps \chi_i(t-\tau_i)+A\hat{x}_i(t)-A\hat{\zeta}_{i1}(t)\\
u_i(t)&=C_{i,h}\xi_i(t)-\rho D_{i,h}K_\eps \chi_i(t),
\end{system}
\end{equation}

Then, we have the following theorem for \emph{scalable} output synchronization of heterogeneous discrete-time MAS in presence of input saturation.

\begin{theorem}\label{thm-dis}
	Consider a heterogeneous network of $N$ agents \eqref{hete_sys} with local information \eqref{local} satisfying Assumption \ref{ass2} and a given $\bar{\tau}$. Let the associated network communication be given by \eqref{zeta-d}. Let $\mathbb{G}^N$ be the set of network graphs as defined in Definition \ref{def1}.
	
	Then  the scalable output synchronization problem as stated in Problem \ref{prob_sync} is solvable.
	In particular, there exist a $\rho^*>0.5$ and for any fixed $\rho>\rho^*$, there exists an $\eps^*$ such that for any $\eps \in (0,\eps^*]$, dynamic protocol given by \eqref{pscp1final} and \eqref{pscp3} solves the scalable output synchronization problem for any $N$ and any graph
	$\mathcal{G}\in\mathbb{G}^N$.
\end{theorem}

\begin{proof}[Proof of Theorem \ref{thm-dis}]	
	Let $\bar{x}_i^o(t)=\bar{x}_i(t)-\bar{x}_N(t)$, $y_i^o(t)=y_i(t)-y_N(t)$, $\hat{x}_i^o(t)=\hat{x}_i(t)-\hat{x}_N(t)$, and $\chi_i^o(t)=\chi_i(t)-\chi_N(t)$. Then, we have 
	\[
	\begin{system*}{ll}
	{\bar{x}}_i^o(t+1)&=A\bar{x}_i^o(t)+B(v_i(t-\tau_i)-v_N(t-\tau_N)\\
	&\hspace{3cm}+\psi_i(t)-\psi_N(t)),\\
	{y}_i^o(t)&=C\bar{x}_i^o(t),\\
	\bar{\zeta}_i(t)&=\zeta_i(t)-\zeta_N(t)=\frac{1}{1+d_{in}(i)}\sum_{j=1}^{N-1}{\ell}_{ij}{y}_j^o(t),\\
	{\hat{x}}_i^o(t+1)&=A\hat{x}_i^o(t)+B(\hat{\zeta}_{i2}(t)-\hat{\zeta}_{N2}(k))+F(\bar{\zeta}_i^d-C\hat{x}_i^o)\\
	{\chi}_i^o(t+1)&=A\chi_i^o(t)+B(v_i(t-\tau_i)-v_N(t-\tau_N))+A\hat{x}_i^o(t)\\
	&\hspace{3cm}-\frac{1}{1+d_{in}(i)}A\sum_{j=1}^{N-1}{\ell}_{ij}{\chi}_j^o(t)\\
	\end{system*}
	\]
	
We define
\[
\begin{system*}{ll}
	\tilde{x}(t)&=\begin{pmatrix}
	\bar{x}_1^o(t)\\ \vdots\\ \bar{x}_{N-1}^o(t)
	\end{pmatrix},\hat{x}(t)=\begin{pmatrix}
	\hat{x}_1^o(t)\\ \vdots\\ \hat{x}_{N-1}^o(t)
	\end{pmatrix},\chi(t)=\begin{pmatrix}
	\chi_1^o(t)\\ \vdots\\ \chi_{N-1}^o(t)
	\end{pmatrix},\\
	\bar{x}^\tau(t)&=\begin{pmatrix}
	\bar{x}_1^o(t-\tau_1)\\ \vdots\\ \bar{x}_{N-1}^o(t-\tau_{N-1})
	\end{pmatrix},\chi^\tau(t)=\begin{pmatrix}
	\chi_1^o(t-\tau_1)\\ \vdots\\ \chi_{N-1}^o(t-\tau_{N-1})
	\end{pmatrix},\\
	\psi(t)&=\begin{pmatrix}
	\psi_1(t)\\ \vdots\\ \psi_N(t)\end{pmatrix},\omega(t)=\begin{pmatrix}
	\omega_1(t)\\ \vdots\\ \omega_N(t)\end{pmatrix}.
\end{system*} 
\]

	Then we have the following closed-loop system
	\begin{equation}
\begin{system*}{ll}
\tilde{x}(t+1)=&(I\otimes A) \tilde{x}(t)-\rho (I\otimes BK_\eps)\chi^\tau(t)+(\Pi\otimes B)\psi(t)\\
\hat{x}(t+1) =& I\otimes (A-FC)\hat{x}(t)-\rho[(I-\tilde{D})\otimes B K_\eps]\chi^\tau(t)\\
&\qquad\qquad+[(I-\tilde{D})\otimes FC]\tilde{x}(t)\\
\chi(t+1) =& [\tilde{D}\otimes A]\chi(t)-\rho(I\otimes BK_\eps)\chi^\tau(t)+(I\otimes A)\hat{x}(t)
\end{system*}
\end{equation}
where $\Pi=\begin{pmatrix}
I&-\mathbf{1}
\end{pmatrix}$ and $\tilde{D}=[\tilde{d}_{ij}]\in \R^{(N-1)\times (N-1)}$, and $\tilde{d}_{ij}=d_{ij}-d_{Nj}$.  We have that the eigenvalues of $\tilde{D}$ are equal to the eigenvalues of $D$ unequal to $1$ (see \cite{wang-saberi-yang}).

By defining $e(t)=\tilde{x}(t)-\chi(t)$ and $\bar{e}(t)=((I-\tilde{D})\otimes I)\tilde{x}(t)-\hat{x}(t)$, we can obtain  
	\begin{equation}\label{newsystem4}
\begin{system*}{ll}
\tilde{x}(t+1)&=(I\otimes A) \tilde{x}(t)-\rho(I\otimes BK_\eps)\tilde{x}^\tau(t)+\rho(I\otimes BK_\eps)e^\tau(t)\\
&\hspace{4.7cm}+(\Pi\otimes B)C_s\omega(t)\\
\bar{e}(t+1)&=I\otimes (A-FC)\bar{e}(t)+((I-\tilde{D})\Pi\otimes B) C_s\omega(t)\\
e(t+1)&=(\tilde{D}\otimes A)e(t)+\bar{e}(t)+(\Pi\otimes B) C_s\omega(t)
\end{system*}
\end{equation}
where $e^\tau(k)=\tilde{x}^\tau(t)-\chi^\tau(t)$. We use the critical lemma \ref{hode-lemma-ddelay} for the stability of delayed system.  The proof proceeds in two steps.
	
\textbf{Step d.1:}  First, we prove the stability of system \eqref{newsystem4} without delay.  By combining \eqref{newsystem4} and \eqref{sys-rho}, we have
\begin{multline}\label{newsystem44}
\begin{pmatrix}
{\tilde{x}}(t+1)\\{\bar{e}}(t+1)\\{e}(t+1)\\{\omega}(t+1)
\end{pmatrix}=\left(\begin{array}{cc}
I\otimes A-\rho(I\otimes BK_\eps)&0\\
0&I\otimes (A-FC)\\
0&I\\
0&0
\end{array}\right.
\\
\left.\begin{array}{cc}
\rho(I\otimes BK_\eps)&(\Pi\otimes B) C_s\\
0&((I-\tilde{D})\Pi\otimes B) C_s\\
\tilde{D}\otimes A&(\Pi\otimes B) C_s\\
0&A_s
\end{array}\right)
\begin{pmatrix}
\bar{x}\\e\\\bar{e}\\\omega
\end{pmatrix}
\end{multline}
The eigenvalues of $\tilde{D}\otimes A$ are of the form
$\lambda_i \mu_j$, with $\lambda_i$ and $\mu_j$ eigenvalues of
$\tilde{D}$ and $A$, respectively \cite[Theorem 4.2.12]{horn-johnson}. Since $|\lambda_i|<1$ and
$|\mu_j|\leq 1$, we find $\tilde{D}\otimes A$ is Schur stable. Meanwhile, we have that $A-FC$ is Schur stable. Then we have
	\[
\lim_{t\to \infty}\bar{e}(t)\to 0 \text{ and }\lim_{t\to \infty}e(t)\to 0
\]
Therefore, we have that the dynamics for $e_i(t)$ and $\bar{e}_i(t)$ are asymptotically stable.

According to the above result, for \eqref{newsystem44} we just need to prove the stability of 
\[
{\tilde{x}}(t+1)=[I\otimes (A- \rho BK_\eps)]\tilde{x}(t)
\]
or Schur stability of $A- \rho BK_\eps$. Based on Lemma \ref{lemma-full}, there exist $\rho >0.5$ and $\eps^* >0$ such that $A- \rho BK_\eps$ is Schur stable
for $\eps \in (0,\eps^*]$.

\textbf{Step d.2:} In this step, since we have that dynamics of $e_i(t)$ and $\bar{e}_i(t)$ are asymptotically stable, we just need to prove the stability of
\[
\tilde{x}_i(t+1)= A \tilde{x}_i(t)- \rho  BK_\eps\tilde{x}_i(t-\tau_{i})
\]
for $i=1,\hdots, N$. Similar to the proof of \cite[Theorem 1]{wang-saberi-stoorvogel-grip-yang}, there exists an $\eps^*$ only function of $(C,A,B)$ such that by choosing 
\[
\rho>\frac{1}{2\cos(\bar{\tau}\omega_{\max} )}
\] 	
we can obtain the synchronization result for any $\eps \in (0,\eps^*]$.

\end{proof}

\section{Simulation Results}
In this section, we will illustrate the effectiveness of our protocols with a numerical example for output synchronization of continuous heterogeneous MAS with partial-state coupling in presence of input delays. We show that our protocol design \eqref{pscp1final} is scale-free and it works for any graph with any number of agents.
Consider the agents models \eqref{hete_sys} with:
\begin{equation*}
\begin{system*}{cl}
A_1=\begin{pmatrix}
0&1&0&0\\0&0&1&0\\0&0&0&1\\0&0&0&0
\end{pmatrix}, B_1=\begin{pmatrix}
0&1\\0&0\\1&0\\0&1
\end{pmatrix}, C_1=\begin{pmatrix}
1&0&0&0
\end{pmatrix},C^m_1=I\\
A_i=\begin{pmatrix}
0&1&0\\0&0&1\\0&0&0
\end{pmatrix},B_i=\begin{pmatrix}
0\\0\\1
\end{pmatrix},C_i=\begin{pmatrix}
1&0&0
\end{pmatrix},
C^m_i=I, \text{ for $i=2,4$}
\end{system*} 
\end{equation*}
and for $i=3,5$
\begin{equation*}
\begin{system*}{cl}
A_i=\begin{pmatrix}
-1&0&0&-1&0\\0&0&1&1&0\\0&1&-1&1&0\\0&0&0&1&1\\-1&1&0&1&1
\end{pmatrix},B_i=\begin{pmatrix}
0&0\\0&0\\0&1\\0&0\\1&0
\end{pmatrix},C_i=\begin{pmatrix}
0&0&0&1&0
\end{pmatrix},
C^m_i=I,
\end{system*}
\end{equation*}
Note that $\bar{n}_d=3$, which is the degree of infinite zeros of $(C_2,A_2,B_2)$. We choose $n_q=3$ and matrices $A,B,C$ as following.
\begin{equation*}
\begin{system*}{cl}
A&=\begin{pmatrix}
0&1&0\\0&0&1\\0&-1&0
\end{pmatrix},\quad B=\begin{pmatrix}
0\\0\\1
\end{pmatrix}, \quad C=\begin{pmatrix}
1&0&0
\end{pmatrix}\\
\end{system*}
\end{equation*}
where
$K=\begin{pmatrix}
30&30&10
\end{pmatrix}\T$ and $H=\begin{pmatrix}
6&10&0
\end{pmatrix}$.
We consider two different heterogeneous MAS with different number of agents and different communication topologies to show that the designed protocols are independent of the
communication networks and the number of agents $N$.
\begin{itemize}
	\item \emph{Case $1$:} Consider a MAS with $3$ agents with agent models $(C_i, A_i, B_i)$ for $i \in \{1,\hdots,3\}$, and directed communication topology with $a_{21}=a_{32}=1$, and input delays $\tau_1=0.1,\tau_2=0.2, \tau_3=0.3$.
	\item \emph{Case $2$:} In this case, we consider a MAS with $3$ agents with agent models $(C_i, A_i, B_i)$ for $i \in \{1,\hdots,5\}$ and directed communication networks with $a_{15}=a_{21}=a_{32}=a_{43}=a_{54}=1$, and input delays $\tau_1=0.1,\tau_2=0.2, \tau_3=0.3, \tau_4=0.4$ and $\tau_5=0.5$.
\end{itemize}
\begin{figure}[t]
	\includegraphics[width=9cm, height=8cm]{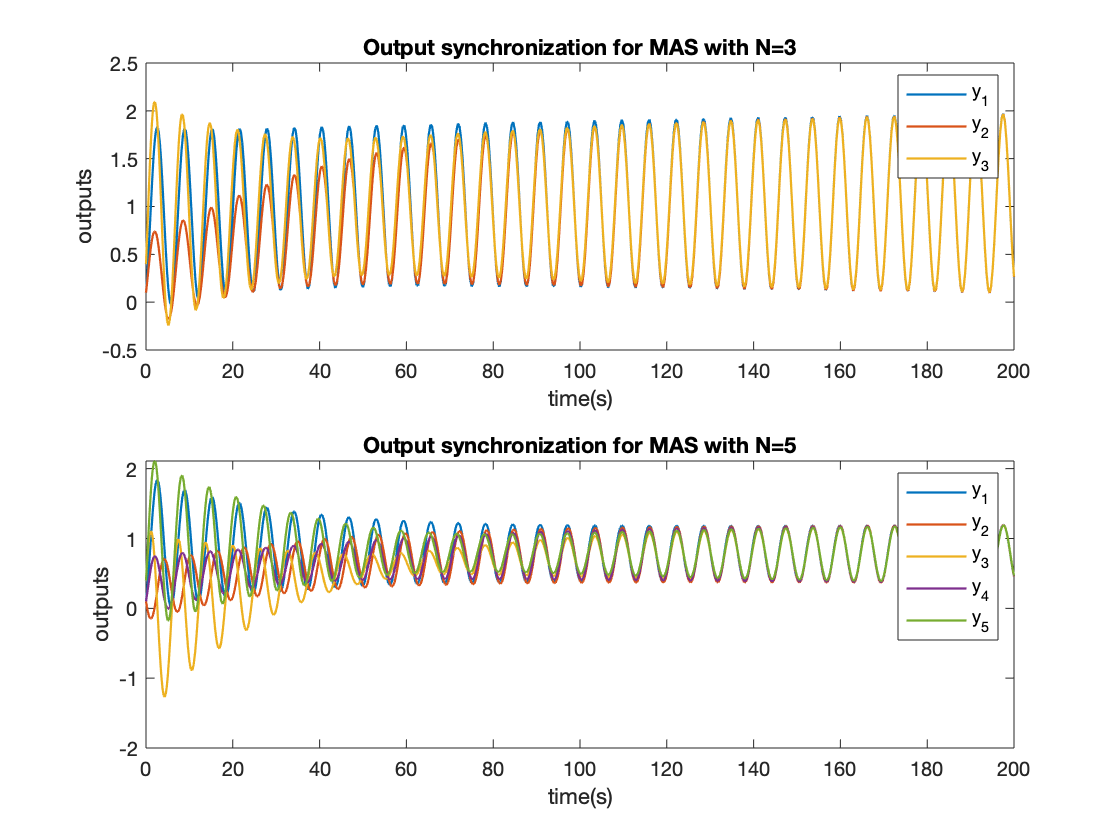}
	\centering
	\caption{Output synchronization for MAS with $N=3$ and $N=5$ agents}\label{results}
\end{figure} 
%
%

\section{Appendix}
\subsection{Stability of delayed continuous-time systems}
Following lemma from \cite{consensus-identical-delay-c} is a classical results in the study of stability of continues-time systems.
\begin{lemma}\label{hode-lemma-ddelay-con}
	Consider a linear time-delay system
	\begin{equation}\label{hode-lemma-system-c}
	\dot{x}(t)=Ax(t)+\sum_{i=1}^{m}A_{i}x(t-\tau_{i}),
	\end{equation}
	where $x(t)\in\R^{n}$ and $\tau_{i}\in\mathbb{R}$. Assume that 
	$A+\sum_{i=1}^{m}A_{i}$ is Hurwitz stable. Then,
	\eqref{hode-lemma-system-c} is asymptotically stable for
	$\tau_1,\ldots,\tau_N\in[0,\bar{\tau}]$ if
	\[
	\det[j\omega I-A- \sum_{i=1}^{m}e^{-j\omega\tau_i}A_{i}]\neq 0,
	\]
	for all $\omega\in\R$, and for all
	$\tau_1,\ldots,\tau_N\in[0,\bar{\tau}]$.
\end{lemma}


\subsection{Stability of delayed discrete-time systems}
We also recall the following lemma from \cite{zhang-saberi-stoorvogel-continues-discrete} for stability of discrete-time systems.
\begin{lemma}\label{hode-lemma-ddelay}
	Consider a linear time-delay system
	\begin{equation}\label{hoded-system1}
	x(t+1)=Ax(t)+\sum_{i=1}^{m}A_{i}x(t-\tau_{i}),
	\end{equation}
	where $x(t)\in\R^{n}$ and $\tau_{i}\in\N^+$. Suppose
	$A+\sum_{i=1}^{m}A_{i}$ is Schur stable. Then, \eqref{hoded-system1}
	is asymptotically stable if
	\[
	\det[e^{j\omega}I-A-
	\sum_{i=1}^{m}e^{-j\omega\tau_i^r}A_{i}]\neq 0,
	\]
	for all $\omega\in[-\pi,\pi]$ and for all $\tau_i\in\overline{[0,\bar{\tau}]}$ \text{ for} ($i=1,\ldots, N$).
\end{lemma}

\subsection{Robustness of low-gain}
Now we recall the following lemmas for the robustness of low-gain design from \cite{lee-kim-shim, consensus-identical-delay-c}.
\begin{lemma}\label{low-gain}
	$A- \rho BB\T P_\eps$ is Hurwitz stable for any $\rho\in \{s\in \mathbb{C} | \re (s) \ge
	\frac{1}{2}\}$ where $P_\eps$ is the unique positive definite solution of
	 \begin{equation}
	 A\T P_{\eps} + P_{\eps} A -  P_{\eps} BB\T
	 P_{\eps} + \eps I = 0 .
	 \end{equation}
\end{lemma}

\begin{lemma}\label{lemma-full}
	Consider a linear uncertain system,
	\begin{equation}\label{linunc}
		x(t+1)=Ax(t)+\lambda Bu(t), \qquad x(0)=x_{0},
	\end{equation}
	where $\lambda\in \mathbb{C}$ is unknown. Assume that $(A,B)$ is
	stabilizable and $A$ has all its eigenvalues in the closed unit
	disc. A low-gain state feedback $u=F_{\delta}x$ is constructed,
	where 
	\begin{equation}\label{fdelta}
		F_{\delta}=-(B\T P_{\delta}B+I)^{-1}B\T P_{\delta}A, 
	\end{equation}
	with $P_{\delta}$ being the unique positive definite solution of the
	$H_{2}$ algebraic Riccati equation,
	\begin{equation}\label{hoded-ARE-sept-17}
		P_\delta=A\T P_\delta A + \delta I - A\T P_\delta B(B\T P_\delta B+I)^{-1}B\T P_\delta A. 
	\end{equation}
	Then, $A+\lambda BF_{\delta}$ is Schur stable for any $\lambda\in \C$
	satisfying,
	\begin{equation}\label{hoded-stablefield}
		\lambda\in \Omega_{\delta}:=\left \{ z\in \mathbb{C}: \left |
		z-\left (1+\tfrac{1}{\gamma_{\delta}}\right ) \right
		|<\tfrac{\sqrt{1+\gamma_{\delta}}}{\gamma_{\delta}}\right \}, 
	\end{equation}
	where $\gamma_{\delta}=\lambda_{\max}(B\T P_{\delta}B)$.  As
	$\delta\rightarrow 0$, $\Omega_{\delta}$ approaches the set
	\[
	H_1:=\{z\in \C :\re z >\tfrac{1}{2}\}
	\]
	in the sense that any compact subset of $H_1$ is contained in
	$\Omega_{\delta}$ for a $\delta$ small enough.
\end{lemma}

\bibliographystyle{plain}
\bibliography{referenc}
\end{document}